\documentclass[11pt]{article}  

\usepackage{amssymb}
\usepackage{amsmath}
\usepackage{amsthm}
\usepackage{color}
\usepackage{colortbl}
\usepackage{graphicx}
\usepackage{hyperref}
\usepackage{fullpage}
\usepackage{algpseudocode}

\newcommand{\I}{\mathcal{I}}

\newtheorem{theorem}{Theorem}
\newtheorem{claim}[theorem]{Claim}
\newtheorem{lemma}[theorem]{Lemma}

\usepackage{lmodern}
\usepackage[T1]{fontenc}

\begin{document}
\title{A note on Cunningham's algorithm for matroid intersection}
\author{Huy L. Nguy\~{\^e}n}
\maketitle

\abstract{In the matroid intersection problem, we are given two matroids of rank $r$ on a common ground set $E$ of $n$ elements and the goal is to find the maximum set that is independent in both matroids. In this note, we show that Cunningham's algorithm for matroid intersection can be implemented to use $O(nr\log^2(r))$ independent oracle calls.}

\section{Introduction}

Consider two matroids $(E, \I_1)$ and $(E, \I_2)$ with rank functions $r_1, r_2$. Let $n=|E|$ and $r$ be the maximum size of an independent set of either matroids. The goal of the matroid intersection problem is to find the maximum set $I\in \I_1\cap \I_2$. Let $T_{ind}$ be the running time of the independent oracle for either matroid. We would like an algorithm that makes as few calls to the independent oracle as possible.

A basic version of Cunningham's algorithm~\cite{Cunningham86} is as follows. We start with an empty solution $S$ and iteratively increase its size by finding augmenting paths. For any $S\in \I_1 \cap \I_2$, an exchange graph $D(S)$ is the directed bipartite graph with bipartition $S$ and $E\setminus S$ such that $(y,x)$ is an arc if $S-y+x\in \I_1$ and $(x,y)$ is an arc if $S-y+x\in \I_2$. Define $X_1 = \{x\not\in S:S+x\in \I_1\}$ be the set of sources and $X_2 = \{x\not\in S:S+x\in \I_2\}$ be the set of sinks.

The algorithm finds a shortest path $P$ from $X_1$ and $X_2$ and replace $S$ with $S\Delta P$ (the symmetric difference between $S$ and $P$).

The sets $X_1$ and $X_2$ can be found in time $O(n T_{ind})$. If $X_1\cap X_2 \ne \emptyset$ then we can simply add those elements to $S$. Thus, we focus on the case $X_1\cap X_2 = \emptyset$. In this case, the running time of the iteration depends on the length of the shortest path, which is bounded by the following result of Cunningham.

\begin{lemma}[\cite{Cunningham86}]
Let $p$ be the maximum size of a common independent set and let $S$ be a common independent set that is not maximum. There exists an augmenting path of length at most $2|S|/(p-|S|)+2$.
\end{lemma}

It is usually stated that a more sophisticated version of Cunningham's algorithm uses $\tilde{O}(nr^{3/2})$ oracle calls. In this note, we will show that even the basic version can be implemented using $\tilde{O}(nr)$ oracle calls.

\section{Implementation of Cunningham's algorithm}

In this section we describe an implementation of Cunningham's algorithm. We find a shortest augmenting path using BFS from $X_1$. The algorithm constructs the graph adaptively as needed for BFS. In every step, given the set $V_d$ of vertices at distance $d$ from the sources, the algorithm needs to find their neighbors that are not reachable before. Those neighbors will be all elements at distance $d+1$ from the sources. Let $A\subseteq S$ and $B\subseteq V\setminus S$ be the elements that are reachable so far. Consider 2 cases depending on the parity of $d$.

\paragraph{Case 1.} Distance $d$ is even. Consider a vertex $v\in V_d \subseteq E\setminus S$. If $v\in X_2$ then the algorithm has found a path from $X_1$ to $X_2$. If not, we need to find $U = \{u\in S\setminus A : S-u+v\in \I_2\}$. 

We describe an algorithm for finding $U$ in time $O((|S|+T_{ind})(|U|+1)\log|S|)$. Initialize $U=\emptyset$.
Consider an ordering $s_1, s_2,\ldots, s_k$ of $S$ such that $A$ form a prefix of the ordering. We use binary search to find the minimum $i$ such that $\{s_1,\ldots, s_i\}\cup\{v\} \not\in \I_2$. If $i\le |A|+|U|$ then algorithm finishes since we have found all elements of the unique circuit of $S\cup\{v\}$. If not, it must be the case that $s_i$ belongs to the unique circuit of $S\cup\{v\}$ and it is an element in $U$. Update $U\gets U\cup\{s_i\}$ and then swap $s_i$ with $s_{|A|+|U|}$. Repeat this algorithm until all elements of $U$ are found. 

\paragraph{Case 2.} Distance $d$ is odd. We need to find $T=\{t \in E\setminus S\setminus B : S-v+t \in \I_1$ for some $v\in V_d\}$. We describe an algorithm for finding $T$ in time $O(n \log |S| (T_{ind}+|S|))$. For each element $u\in E\setminus S\setminus B$, we check if $S\setminus V_d\cup\{u\}\in \I_1$. If so, $u\in T$ because the unique circuit of $S\cup\{u\}$ contains some element in $V_d$. In fact, one can find one $v\in V_d$ such that $S-v+u\in \I_1$ by binary search: we find the minimum $i$ such that the union of $S\cup\{u\}\setminus V_d$ and the first $i$ elements of $V_d$ is not in $\I_1$. The element we are looking for is the $i$th element in $V_d$.

\begin{claim}
The running time of the algorithm is $O(nr\log^2 r (T_{ind}+r))$.
\end{claim}
\begin{proof}
Next we analyze the running time of the whole algorithm across all iterations of BFS. In case 1, each element in $S$ can be in $U$ at most once when its status changes from unreachable to reachable. Thus, the total running time of case 1 is $O(n\log r (T_{ind}+r))$. The running time of case 2 depends on the length of the shortest path. If the length of the shortest path is $\ell$, the total running time of case 2 is $O(\ell n\log r (T_{ind}+r))$. 

By the bound on the length of the shortest augmenting path mentioned above, the total length of the augmenting paths is $O(r\log r)$.

Thus, the total running time is $O(nr\log^2 r (T_{ind}+r))$.
\end{proof}
\bibliographystyle{plain}
\bibliography{matroid}
\end{document}